\documentclass[reqno,a4paper]{amsart}

\usepackage{cite,url}

\newcommand{\RR}{\mathbb{R}}

\newcommand{\CC}{\mathbb{C}}
\newcommand{\NN}{\mathbb{N}}

\newcommand{\cH}{\mathcal{H}}
\newcommand{\cG}{\mathcal{G}}

\newcommand{\cB}{\mathcal{B}}
\newcommand{\cL}{\mathcal{L}}
\newcommand{\cK}{\mathcal{K}}

\newcommand{\cV}{\mathcal{V}}
\newcommand{\cE}{\mathcal{E}}
\newcommand{\cU}{\mathcal{U}}

\newtheorem{theorem}{Theorem}
\newtheorem*{theorem*}{Theorem}
\newtheorem{prop}[theorem]{Proposition}
\newtheorem{lemma}[theorem]{Lemma}
\newtheorem{corol}[theorem]{Corollary}

\theoremstyle{remark}

\newtheorem{example}[theorem]{Example}

\newcommand{\slim}{\mathop{\text{s-lim}}\limits}

\DeclareMathOperator{\spec}{spec}
\DeclareMathOperator{\dom}{dom}
\DeclareMathOperator{\ran}{ran}
\DeclareMathOperator{\Id}{Id}
\DeclareMathOperator{\indeg}{indeg}
\DeclareMathOperator{\outdeg}{outdeg}

\begin{document}

\title[Dimension reduction for self-adjoint extensions]{Unitary dimension reduction for a class of self-adjoint extensions with applications to graph-like structures}

\author{Konstantin Pankrashkin}

\address{Laboratoire de math\'ematiques d'Orsay (UMR 8628)\\
Universit\'e Paris-Sud 11, B\^atiment 425\\
91405 Orsay Cedex, France}

\email{konstantin.pankrashkin@math.u-psud.fr}
\urladdr{http://www.math.u-psud.fr/~pankrash/}

\begin{abstract}
We consider a class of self-adjoint extensions using the boundary triple technique.
Assuming that the associated Weyl function has the special form
$M(z)=\big(m(z)\Id-T\big) n(z)^{-1}$ with a bounded self-adjoint operator
$T$ and scalar functions $m,n$ we show that there exists
a class of boundary conditions such that the spectral problem for the associated self-adjoint
extensions in gaps of a certain reference operator admits a unitary reduction
to the spectral problem for $T$. As a motivating example we consider
differential operators on equilateral metric graphs, and we describe a class of boundary
conditions that admit a unitary reduction to generalized discrete laplacians.
\end{abstract}



\maketitle

\section{Introduction}

The present work is motivated by the study of the relationship between discrete operators
on graphs and differential operators on metric graphs (quantum graphs),
see \cite{GS,QG1,EDAG, AGA,Ku}. Let us recall the basic notions and introduce an illustrative example.

Let $G$ be a countable graph, the sets of the vertices and of the edges of $G$
will be denoted by $\cV$ and $\cE$, respectively, and multiple edges and self-loops
are allowed. For an edge $e\in \cE$ we denote by $\iota{e}\in \cV$ its initial
vertex and by $\tau{e}\in\cV$ its terminal vertex.
For a vertex $v$, the number of outgoing edges and the number of ingoing edges
will be denoted by $\outdeg v$ and $\indeg v$, respectively, and
the degree of $v$ is $\deg v:=\indeg v+\outdeg v$.
In what follows we assume that the degrees of the vertices
are uniformly bounded and that there are no isolated vertices, i.e.
$1\le\deg v\le N$ for all $v\in \cV$. Introduce the discrete Hilbert space
\[
l^2(G):=\big\{ f:\cV\to\CC:\, \|f\|^2=\sum_{v\in\cV} \deg v |f(v)|^2<+\infty\big\}
\]
and the weighted adjacency operator $\Delta$ in $l^2(G)$,
\begin{equation}
       \label{eq-delta}
(\Delta f)(v)=\dfrac{1}{\deg v}\Big(\sum_{e:\iota v=e} f(\tau v) + \sum_{e:\tau e=v} f(\iota e)\,\Big).
\end{equation}
Numerous works treat the relationship between the properties of $\Delta$ and $G$, see e.g.
\cite{chung} and references therein.

Let us now introduce a continuous Laplacian on $G$. Consider the Hilbert space
$\cH:=\bigoplus_{e\in\cE}\cH_e$, $\cH_e=L^2(0,1)$,
and the operator $\Lambda$, $\Lambda(f_e)=(-f''_e)$, acting on the functions
$f=(f_e)\in H^2(0,1)$ satisfying the so-called standard boundary conditions:
\begin{gather*}
f_e(1)=f_b(0) \text{ for all } b,e\in \cE \text{ with } \iota{b}=\tau{e}
\text{ (=continuity at each vertex)},\\
\sum_{e:\iota{e}=v} f'_e(0)-\sum_{e:\tau{e}=v} f'_e(1)=0.
\end{gather*}
It is known that  $\Lambda$ is self-adjoint and that its spectrum is closely related
with the spectrum of $\Delta$: denoting $\sigma_D=\{(\pi n)^2:\, n\in\NN\}$ one has
the relationship
\begin{equation}
        \label{eq-dual}
\spec_j \Lambda\setminus \sigma_D= \{z\notin\sigma_D: \cos\sqrt{z}\in\spec_j \Delta\},
\quad
j\in\{\text{p},\text{pp},\text{disc},\text{ess},\text{ac},\text{sc}\}.
\end{equation}
For $j\in\{\text{p},\text{disc},\text{ess}\}$ this was proved,
for example, in \cite{vB85} for finite graphs and in \cite{C97} for infinite graphs.
In \cite{BGP08} the result was obtained for the first time
for all types of spectra, and the work~\cite{P08} used the results of \cite{BGP08}
to prove a similar result for continuous Laplacians with more general boundary conditions.
We refer e.g. to \cite{vblub,RC,dn2,sn1,MH,KP08,KP09,BGP07,KuPo,KP06,KP06b,P09,exdual}
for generalizations to more general differential operators and for the analysis
of particular configurations. 
The aim of the present paper is to improve the relation \eqref{eq-dual}.
If $\Omega$ is a Borel set in $\RR$ and $A$ is a selfadjoint operator,
denote by $A_\Omega$ the part of $A$ in $\Omega$, i.e. 
$A_\Omega=A1_\Omega(A)$ considered as an operator in $\ran 1_\Omega(A)$;
here $1_\Omega(A)$ is the spectral projector of $A$ onto $\Omega$.
A simple corollary of Theorem \ref{th2} below is the following
\begin{prop}\label{th0}
Denote $\eta(z):=\cos\sqrt{z}$, then for any interval $J\subset\RR\setminus\sigma_D$
the operator $\Lambda_J$ is unitarily equivalent to the operator $\eta^{-1}\big(\Delta_{\eta(J)}\big)$.
\end{prop}

It was noted by the author in \cite{KP06} that the operator $\Lambda$ can be studied
at an abstract level using the language
of boundary triples and self-adjoint extensions \cite{GG,DM91,BGP08}. 
Let $S$ be a closed densely defined symmetric operator
in a separable Hilbert space $\cH$ with the domain $\dom S$. Assume that $S$ has equal deficiency indices,
i.e. that $\dim\ker(S^*+i)=\dim\ker(S^*-i)$. A \emph{boundary triple} for $S$ consists of a Hilbert space $\cG$
and two linear maps $\Gamma,\Gamma':\dom S\to\cG$ satisfying the following two conditions:
\begin{itemize}
\item $\langle f,S^*g\rangle-\langle S^*f,g\rangle =\langle \Gamma f,\Gamma'g\rangle-\langle \Gamma'f,\Gamma g\rangle$
for all $f,g\in\dom S^*$,
\item the application $(\Gamma,\Gamma'):\dom S^*\ni f \mapsto (\Gamma f,\Gamma' f)\in \cG\oplus\cG$
is surjective.
\end{itemize}
We will consider the two distinguished self-adjoint extensions of $S$:
\begin{equation}
     \label{eq-hh0}
H^0:=S^*|_{\ker \Gamma} \text{ and } H:=S^*|_{\ker\Gamma'}.
\end{equation}
It is known \cite{DM91} that for any two self-adjoint extensions $H^0$ and $H$
satisfying $\dom H\cap\dom H^0=\dom S$ ($H$ and $H^0$ are then called disjoint)
one can find a boundary triple $(\cG,\Gamma,\Gamma')$ such that \eqref{eq-hh0} holds.
An essential role in the analysis of the self-adjoint
extensions is played by the so-called Weyl function $M(z)$ which is defined as follows.
For $z\notin\spec H^0$ consider the operator $\gamma(z):=\big(\Gamma|_{\ker(S^*-z)}\big)^{-1}$
which is a linear topological isomorphism between $\cG$ and $\ker(S^*-z)\subset\cH$,
then the map $\CC\setminus\spec H^0\ni z\mapsto \gamma(z)\in \cL(\cG,\cH)$ (called $\gamma$-field) is holomorph.
The operator function $\CC\setminus\spec H^0\ni z\mapsto M(z):=\Gamma'\gamma(z)\in \cL(\cG)$
is called the \emph{Weyl} function associate with the boundary triple.
Outside $\spec H^0\cup\spec H$ the Krein resolvent formula holds,
$(H^0-z)^{-1}-(H-z)^{-1}=\gamma(z)M(z)^{-1}\gamma(z)^*$,
and we have the relation \cite{DM91,BGP08}
\begin{equation}
       \label{eq-mspec}
\spec_j H\setminus\ \spec H^0=\big\{
z\notin\spec H^0:\, 0\in\spec_j M(z)
\big\}, \quad
j\in\{\text{p},\text{disc},\text{ess}\}.
\end{equation}
Numerous papers were devoted to the question whether
one can detalize the relation \eqref{eq-mspec} and to recover, for example,
the singular or the absolutely continuous spectrum of $H$ in terms
on the spectral properties of $M$, see e.g. \cite{ABMN05,jfb,BMN02,BGP08, DM91,DM} and references there-in.
Our main result contributes this direction
and concerns Weyl functions of a special form.
\begin{theorem}\label{th-main}
Assume that the Weyl function $M$ has the form
\begin{equation}
        \label{eq-special}
M(z)=\dfrac{m(z)\Id-T}{n(z)}
\end{equation}
where
\begin{itemize}
\item $T$ is a bounded self-adjoint operator in $\cG$,
\item $m$ and $n$ are scalar functions which
are holomorph outside $\spec H^0$.
\end{itemize}

Assume that there exists a spectral gap $J:=(a_0,b_0)\subset \RR\setminus \spec H^0$
such that $m$ and $n$ admit a holomorph continuation to $J$, are both real-valued in $J$,
that $n\ne 0$ in $J$, and that $m(J)\cap \spec T\ne\emptyset$, then 
\begin{itemize}
  \item[(a)] there exists an interval $K$ containing $m^{-1}(\spec T)\cap J$
  such that $m:K\to m(K)$ is a bijection; denote by $\mu$ the inverse function;
  \item[(b)] the operator $H_J$ is unitarily equivalent to $\mu(T_{m(J)})$. 
\end{itemize}
\end{theorem}
As was shown in \cite{KP06}, the analysis of the above operator $\Lambda$ can be put
into the framework of boundary triples: the associated Weyl function in suitable coordinates
has the requested form $M(z)=\big(\Delta-\cos\sqrt{z}\,\Id\big) \sqrt z / \sin\sqrt z$,
and Proposition \ref{th0} becomes a simple corollary of Theorem \ref{th-main}.
We recall these constructions and generalize the above example in Section \ref{sec-gr}.

Theorem \ref{th-main} shows that the spectral analysis of $H$ in the interval $J$
is equivalent to the spectral analysis of the operator $T$ on a ``smaller'' space $\cG$,
and this fact can be considered as a dimension reduction.
Note that for $n=\text{const}\ne 0$ Theorem \ref{th-main}
is actually proved in~\cite{ABMN05}: it is not stated explicitly, but the proof of Theorem 4.4 in \cite{ABMN05}
contains the result, and we are adapting their scheme of proof to the case of non-constant $n$.
The main difference comes from the fact that for constant $n$ the function
$m$ is strictly increasing, while this is no more true for general $n$, which
brings some additional difficulties.
Note that the results of \cite{ABMN05} are suitable for the analysis of operators
that can be represented as direct sums of operators with deficiency indices $(1,1)$,
but this does not cover the above example with the continuous graph laplacian.

We emphasize that the condition $m(J)\cap\spec T\ne\emptyset$ in Theorem \ref{th-main}
is just to avoid some pathologies in the notation and this does not bring any restriction.
If $m(J)\cap\spec T=\emptyset$, then by \eqref{eq-mspec} the operator $H$ has no spectrum in $J$,
and the assertion (b) still holds formally, as the both operators are defined on the zero space.

Note that as an obvious corollary of Theorem \ref{th-main} we have
the following assertion obtained already in the author's joint work \cite[Theorem 3.16]{BGP08} by a different method:
\begin{corol}
For any $x\in J$ and any $j\in\{\text{p},\text{pp},\text{disc},\text{ess},\text{ac},\text{sc}\}$
the assertions
\begin{itemize}
\item $x\in\spec_j H$,
\item $m(x)\in\spec_j T$
\end{itemize}
are equivalent.
\end{corol}

\section{Proof of the unitary equivalence}

This section is devoted to the proof of Theorem \ref{th-main}.

\subsection{Operator-valued mesures}

In what follows by $\cB(\RR)$ we denote the algebra of Borel subsets of $\RR$, and by
$\cB_b(\RR)$ its subalgebra consisting of the bounded Borel subsets.
If $\cH$ and $\cH'$ are Hilbert spaces, then $\cL(\cH,\cH')$ stands for the space of bounded linear operators
from $\cH$ to $\cH'$, and $\cL(\cH):=\cL(\cH,\cH)$.
A mapping $\Sigma:\cB_b(\RR)\to \cL(\cH)$ is called an \emph{operator-valued measure} (in $\cH$)
if it is $\sigma$-additive
with respect to the strong convergence and if $\Sigma(B)=\Sigma(B)^*\ge 0$ for all $B\in \cB_b(\RR)$.
An operator-valued measure $\Sigma$ is called \emph{bounded} if extends by $\sigma$-additivity to a map
$\cB(\RR)\to \cL(\cH)$.
A bounded operator-valued measure $\Sigma$ is called \emph{orthogonal} if it satisfies two additional conditions:
$\Sigma(B_1\cap B_2)=\Sigma(B_1)\Sigma(B_2)$ for all $B_1,B_2\in\cB(\RR)$
and $\Sigma(\RR)=\Id$.

Let $\cH_1$, $\cH_2$ be Hilbert spaces, $K:\cH_2\to\cH_1$ be a bounded linear operator,
and $\Sigma_1$ be a bounded operator-valued spectral measure in $\cH_1$,
then the mapping $\Sigma_2:\cB(\RR)\ni B\mapsto \Sigma_2(B):=K^*\Sigma_1(B)K\in\cL(\cH_2)$
is a bounded operator-valued measure in $\cH_2$ which is called a \emph{dilation} of $\Sigma_1$.
This dilation is \emph{orthogonal} if the above representation holds
with a unitary operator $K$ and is called \emph{minimal} if the closed linear span of the subspaces $\Sigma_1(B)\ran K$,
$B\in\cB(\RR)$, coincides with $\cH_1$. If a bounded operator-valued measure is an orthogonal dilation of another
bounded operator-valued measure, then these two measures are called \emph{unitarily
equivalent}. Note that the spectral measure of a self-adjoint operator is always
an orthogonal operator-valued measure. The following assertion is well known, see e.g. \cite[Chapter 4]{paul}
or \cite{MM}.

\begin{theorem}[Generalized Naimark's dilation theorem]\label{naimark}
Any bounded operator-valued measure $\Sigma$ can be represented as a minimal dilation
of an orthogonal operator-valued measure $\Sigma^0$, and
$\Sigma^0$ is called a \emph{minimal orthogonal operator-valued measure associated with $\Sigma$}. 
If a bounded operator-valued measure can be represented as a minimal orthogonal dilation
of two different orthogonal operator-valued measures, then these two orthogonal
operator-valued measures are unitarily equivalent.
\end{theorem}

Let us recall some tools that allows one to obtain some information on the spectral measures
for self-adjoint extensions using the Weyl functions.

Let $\CC_+:=\{z\in\CC:\, \Im z>0\}$ and $\cH$ be a Hilbert space.
A map $\CC_+\ni z\mapsto F(z)\in\cL(\cH)$ is called an (operator-valued) \emph{Herglotz function}
on $\cH$ if $\Im F(z)\ge 0$ for all $z\in\CC_+$. To each Herglotz function $F$ on $\cH$
one can associate a uniquely defined bounded operator-valued measure (bounded Herglotz measure),
in $\cH$, which we denote by $\Sigma^0_F$, 
and two non-negative operators $C_1$ and $C_2$ on $\cH$ such that
\[
F(z)=C_0+C_1 z+\int_{\RR} \dfrac{1+tz}{t-z} \, \Sigma^0_F(dt)
\text{ for all }z\in\CC_+.
\]
On can introduce another operator-valued measure $\Sigma_F$ (unbounded Herglotz measure)
associated with $F$ by the equality
\[
\Sigma_F(B):=\int_B (1+t^2) \Sigma^0_F(dt), \quad B\in\cB_b(\RR).
\]
This operator-valued measure is unbounded in general, but it can be recovered
from the values $F$ by the explicit Stieltjes inversion formula
\begin{equation}
       \label{eq-stil}
\Sigma_F\big((a,b)\big)=\slim_{\delta\to 0+} \slim_{\varepsilon\to 0+}
\dfrac{1}{\pi}\int_{a+\delta}^{b-\delta} \Im F(x+i\varepsilon)\,dx,
\end{equation}
see \cite{A,A2}.
Note that the Weyl function $M(z)$ defined by a boundary triple is always a Herglotz function
and satisfies $M(\Bar z)=M(z)^*$, see e.g. \cite[Proposition 1.21]{BGP07}.
The following fact is known \cite[Lemma 2.12]{ABMN05}:
\begin{prop} \label{prop-sae}
Let $S$ be a closed densely defined symmetric operator in a Hilbert space $\cH$
with equal deficiency indices, and let $(\cG,\Gamma,\Gamma')$ be an associated boundary triple.
Let $M$ be the associated Weyl function and $H^0$ be the restriction of $S^*$ to $\ker \Gamma$.
Assume that $S$ is simple (i.e. has no invariant subspaces on which it is self-adjoint), then
the spectral measure for $H^0$ is a minimal orthogonal operator-valued measure
associated with the bounded operator-valued Herglotz measure $\Sigma^0_M$ associated with $M$.
\end{prop}

The following proposition combines the above results and provides
a step toward the proof of Theorem \ref{th-main}.

\begin{prop}\label{prop-main} Let the assumptions of Theorem \ref{th-main} be fulfilled,
and let the assertion (a) of Theorem \ref{th-main} hold. Set $N(z):=-M(z)^{-1}$ and let $\Sigma^0_N$
be the associated bounded Herglotz measure. Define its restriction $\Sigma^0_{N,J}$ onto $J$
by $\Sigma^0_{N,J}(B)=\Sigma^0_N(B\cap J)$. If $\Sigma^0_{N,J}$
is a minimal dilation of the spectral measure $E_R$ of the operator $R=\mu\big(T_{m(J)}\big)$,
then the operators $H_J$ and $R$ are unitarily equivalent.
\end{prop}

\begin{proof}
(a) Assume first that $S$ is a simple operator. Introduce the new boundary triple $(\cG,\widetilde\Gamma,\widetilde\Gamma')$
with $\widetilde\Gamma:=-\Gamma'$ and $\widetilde\Gamma':=\Gamma$.
The associated Weyl function is $N(z):=-M(z)^{-1}$, and is hence also a Herglotz one, 
and the operator $H$ becomes then the restriction of $S^*$ to $\ker\widetilde\Gamma$.
By Proposition \ref{prop-sae} one can
represent $\Sigma_N^0$ as a minimal dilation of the spectral measure $E_H$ of
$H$, $\Sigma_N^0(B)=K^* E_H(B) K$, $K\in \cL(\cG,\cH)$, then
\[
\Sigma^0_{N,J}(B)=\Sigma^0_N(B\cap J)=K^* E_H(B\cap J) K
=L^* E_{H,J}(B) L,
\]
where $E_{H,J}$ defined by $E_{H,J}(B)= E_H (B\cap J)$
is considered as an orthogonal measure in $\cH':=\ran E_H(J)$, and
$L= \Pi K$ with $\Pi:\cH\to\cH'$ being the orthogonal projector.
Therefore, $E_{H,J}$ is another minimal orthogonal measure associated with $\Sigma^0_{N,J}$,
hence $E_R$ and $E_{H,J}$ are unitarily equivalent by Naimark's theorem (Theorem \ref{naimark}).
This means that there exists a unitary $U$ such that $E_{H,J}(B)=U^* E_R(B) U$ for all $B\subset J$,
and
\[
H_J=\int_J t\, E_{H,J}(dt)=U^*\int_J t\, E_R(dt)\, U=U^* R U.
\]

(b) If the operator $S$ is not simple, one can decompose the Hilbert space $\cH$ and the operator $S$
into a direct sum
$\cH=\cH_0\oplus\cK$, $S=S_0\oplus L$, such that $L$ is a self-adjoint operator in $\cK$ and $S_0$ is a closed densely
defined \emph{simple} symmetric operator in $\cH_0$ whose deficiency indices
are equal to those for $S$. Moreover, $(\cG,\Bar\Gamma,\Bar\Gamma')$, where $\Bar \Gamma$ and $\Bar\Gamma'$
are the restrictions of $\Gamma$ and $\Gamma'$ respectively to $\dom S^*_0$,
is a boundary triple for $S_0$ with the same Weyl function $M(z)$. Moreover,
one has $H^0=A^0\oplus L$ and $H=A\oplus L$, where $A^0$ is the restriction
of $S_0^*$ to $\ker\Bar\Gamma$ and $A$ is the restriction of $S_0^*$ to $\ker\Bar\Gamma'$.
One has $J\subset \RR\setminus\spec A^0$ and $J \subset \RR\setminus\spec L$, which means
that $H_J$ is unitarily equivalent to $A_J$. Finally, applying the part (a) to the operators $S_0$, $A$ and $A^0$
one shows that $A_J$ is unitarily equivalent to $R$. 
\end{proof}

\subsection{Technical estimates}

In this section we use the notation and the assumptions introduced in Theorem \ref{th-main}
and Proposition \ref{prop-main}. The aim of this section is to calculate the bounded Herglotz measure
$\Sigma^0_N$ associated to $N$ in terms of the spectral measure for the operator $R$.

Denote
\begin{equation}
    \label{eq-stk}
S_T:=[\inf \spec T,\sup \spec T], \quad K:=m^{-1}(S_T)\cap J.
\end{equation}
The following assertion was proved in \cite[Lemma 3.13]{BGP07}:
\begin{lemma} \label{prop-pm}
For any $x\in K$ one has $m'(x)\ne 0$.
\end{lemma}
We will prove below
\begin{lemma}\label{lem4}
The set $K$ is connected.
\end{lemma}

Let $(a,b)\subset J$. By the Stieltjes inversion formula \eqref{eq-stil} one has
\begin{equation}
     \label{eq-sigma0}
\Sigma^0_N\big((a,b)\big)=\slim_{\delta\to 0+}\slim_{\varepsilon\to 0+}
\dfrac{1}{2\pi i} \int_{a+\delta}^{b-\delta} \Big(N(x+i\varepsilon)-N(x-i\varepsilon)\Big)\, dx.
\end{equation}
On the other hand, there holds
\begin{multline}
     \label{eq-m0}
N(x+i\varepsilon)-N(x-i\varepsilon)=\int_\RR \Big(
\dfrac{n(x+i\varepsilon)}{\lambda-m(x+i\varepsilon)}-\dfrac{n(x-i\varepsilon)}{\lambda-m(x-i\varepsilon)}
\Big)\,E_T(d\lambda)\\
=\int_{S_T}\Big(
\dfrac{n(x+i\varepsilon)}{\lambda-m(x+i\varepsilon)}-\dfrac{n(x-i\varepsilon)}{\lambda-m(x-i\varepsilon)}
\Big)\,E_T(d\lambda),
\end{multline}
where $E_T$ is the spectral measure associated with $T$.

For a Borel subset $I$ of $J$ denote
\begin{equation}
    \label{eq-kI}
k_I(\lambda,\varepsilon)=\dfrac{1}{2\pi i}
\int_I
\Big(
\dfrac{n(x+i\varepsilon)}{\lambda-m(x+i\varepsilon)}
-
\dfrac{n(x-i\varepsilon)}{\lambda-m(x-i\varepsilon)}
\Big)dx.
\end{equation}

Our main technical estimate is the following proposition.

\begin{prop}\label{prop-conv} 
Assume that $I=[a,b]\subset J$. For some $\varepsilon_0>0$ there holds
\begin{equation}
          \label{eq-kbd}
\sup_{\substack{\lambda\in S_T\\ \varepsilon\in (0,\varepsilon_0)}}\big|k_I(\lambda,\varepsilon)\big| <+\infty
\end{equation}
and for any $\lambda\in S_T$ one has
\begin{equation}
              \label{eq-klim}
\lim_{\varepsilon\to 0+}k_I(\lambda,\varepsilon)=\begin{cases}
0, & \lambda\notin m\big([a,b]\big),\\
\dfrac{1}{2}\,\mu'(\lambda) n\big(\mu(\lambda)\big),& \lambda\in \big\{m(a),m(b)\big\},\\
\mu'(\lambda)n\big(\mu(\lambda)\big),& \lambda\in m\big((a,b)\big).
\end{cases}
\end{equation}
Here $\mu$ is the inverse to $K\ni x\mapsto m(x)\in m(K)$; this inverse exists
by Lemmas \ref{prop-pm} and \ref{lem4}.
\end{prop}

To prove proposition \ref{prop-conv} let us make some preliminary steps.
\begin{lemma} \label{lem1}
Let $I\subset J$ be a closed segment such that $m'(x)\ne 0$ for $x\in I$.
Then, for some $\varepsilon_0>0$ and for all $x\in I$, $\lambda\in \RR$ and $0<|\varepsilon|<\varepsilon_0$
there holds
\begin{equation}
         \label{eq-res1}
\dfrac{1}{\lambda-m(x+i\varepsilon)}=\dfrac{1}{\lambda -m(x) -i\varepsilon m'(x)}
\cdot\Big(
1 + \varepsilon\, g(x,\lambda,\varepsilon)
\Big),
\end{equation}
where
\[
\sup_{\substack{x\in I,\, \lambda\in \RR\\ 0<|\varepsilon|<\varepsilon_0}} \big|g(x,\lambda,\varepsilon)\big|<+\infty.
\]
\end{lemma}

\begin{proof}

There holds
\begin{equation}
      \label{eq-inv}
\dfrac{1}{\lambda-m(x+i\varepsilon)}=
\dfrac{f(x,\lambda,\varepsilon)}{\lambda-m(x)-i\varepsilon m'(x)}
\end{equation}
with
\begin{equation}
f(x,\lambda,\varepsilon)= \dfrac{\lambda-m(x)-i\varepsilon m'(x)}{\lambda-m(x+i\varepsilon)}=
1+ \dfrac{m(x+i\varepsilon)-m(x)-i\varepsilon m'(x)}{\lambda-m(x+i\varepsilon)}.
\end{equation}

Due to the analyticity of $m$, there exists $C>0$ such that 
\begin{equation}
         \label{eq-est1}
\big|m(x)+i\varepsilon m'(x) - m(x+i\varepsilon)\big|\le C\varepsilon^2 \quad \text{for all }
x\in I, |\varepsilon|<\varepsilon_0.
\end{equation}
On the other hand, denoting $k=\inf_{x\in I} |m'(x)|>0$,
one has $\big|\lambda-m(x)-i\varepsilon m'(x)\big|\ge k|\varepsilon|$.
Therefore, one can find $c>0$ such that
\begin{equation}
        \label{eq-est2}
\big|\lambda-m(x+i\varepsilon)\big|\ge c  |\varepsilon|
\text{ for all  $\lambda\in\RR$, $x\in I$, $|\varepsilon|\le \varepsilon_0$.}
\end{equation}
Using \eqref{eq-est1} and \eqref{eq-est2} one obtains, with $b=C/c>0$, 
\[
\Big|\dfrac{m(x+i\varepsilon)-m(x)-i\varepsilon m'(x)}{\lambda-m(x+i\varepsilon)}\Big|\\
\le b\varepsilon \text{ for all $x\in I$, $\lambda\in\RR$, $0<|\varepsilon|<\varepsilon_0$.} \qedhere
\] 
\end{proof}

\begin{lemma} \label{lem0} The result of proposition \ref{prop-conv} holds under the additional assumption
\[
m'(x)\ne 0 \text{ for all } x\in I.
\]
\end{lemma}

\begin{proof}
Let us take the same $\varepsilon_0$ as in Lemma \ref{lem1}.
Using the representation \eqref{eq-res1} one can write
\begin{multline} 
       \label{eq-k1}
k_I(\lambda,\varepsilon)=\dfrac{1}{2\pi i}\int_{a}^{b}
\Bigg[
\dfrac{n(x+i\varepsilon)\cdot\Big(
1 + \varepsilon \, g(x,\lambda,\varepsilon)
\Big)}{\lambda-m(x)-i\varepsilon m'(x)}\\
-
\dfrac{n(x-i\varepsilon)\cdot\Big(
1 -\varepsilon\, g(x,\lambda,-\varepsilon)
\Big)}{\lambda-m(x)+i\varepsilon m'(x)}
\Bigg]\,dx.
\end{multline}
As $n$ is holomorph, one can write $n(x+i\varepsilon)=n(x)+\varepsilon p(x,\varepsilon)$
with
\[
\sup_{\substack{x\in I\\ |\varepsilon|<\varepsilon_0}} \big|p(x,\varepsilon)\big|<+\infty.
\]
Substituting this representation into \eqref{eq-k1} one obtains
\begin{multline}
k_I(\lambda,\varepsilon)=\underbrace{\dfrac{1}{2\pi i}\int_{a}^{b}
n(x)\Big(
\dfrac{1}{\lambda-m(x)-i\varepsilon m'(x)}
-
\dfrac{1}{\lambda-m(x)+i\varepsilon m'(x)}
\Big)dx}_{=:I_1(\lambda,\varepsilon)}
\\
+
\underbrace{ 
\dfrac{1}{2\pi i}\int_{a}^{b} 
\dfrac{\varepsilon r(x,\lambda, \varepsilon)}{\lambda-m(x)-i\varepsilon m'(x)}\,dx}_{=:I_2(\lambda,\varepsilon)}
+
\underbrace{ 
\dfrac{1}{2\pi i}\int_{a}^{b} 
\dfrac{\varepsilon r(x,\lambda, -\varepsilon)}{\lambda-m(x)+i\varepsilon m'(x)}\,dx}_{=:I_3(\lambda,\varepsilon)}.
\end{multline}
with
\[
r(x,\lambda,\varepsilon):= p(x,\varepsilon)\big(1+\varepsilon g(x,\lambda,\varepsilon)\big) + n(x) g(x,\lambda,\varepsilon).
\]
One has obviously
\[
\sup_{\substack{x\in I,\, \lambda\in \RR\\ 0<|\varepsilon|<\varepsilon_0}} \big|r(x,\lambda,\varepsilon)\big|=:C<+\infty
\]
Denoting
\[
k=\inf_{x\in [a,b]} |m'(x)|>0
\]
one can estimate, for all $\lambda\in\RR$ and $0<|\varepsilon|<1$,
\begin{equation} 
           \label{eq-rbd}
\Big|
\dfrac{\varepsilon r(x,\lambda, \varepsilon)}{\lambda-m(x)+i\varepsilon m'(x)}
\Big|\le \dfrac{R}{k}.
\end{equation}
Therefore, one has
\[
\big|
I_{2,3}(\lambda,\varepsilon)
\big|\le \dfrac{R|b-a|}{2\pi k}
\text{ for all $\lambda\in\RR$ and $0<|\varepsilon|<1$.}
\]
Let us study the expression for $I_1$. By elementary transformations one obtains
\[
I_1(\lambda,\varepsilon)=
\dfrac{1}{\pi}\,\int_{a}^{b}\dfrac{\varepsilon m'(x) n(x)}{\big(\lambda-m(x)\big)^2 + \big(\varepsilon m'(x)\big)^2}\, dx.
\]
Denoting $N:=\sup_{x\in I} \big|n(x)\big|$ one obtains
\begin{multline*}
|I_1|\le \dfrac{N}{\pi} \int_{a}^{b} \dfrac{\big|m'(x)\big|}{\big(\lambda-m(x)\big)^2 + \varepsilon^2 k^2}\,dx\\
=\dfrac{N}{\pi} \Big|
\int_{m(a)}^{m(b)} \dfrac{\varepsilon}{(\lambda -y)^2+\varepsilon^2k^2}\,dy
\le
\dfrac{N}{\pi} \int_{-\infty}^{+\infty} \dfrac{\varepsilon}{y^2+\varepsilon^2k^2}\,dy =\dfrac{N}{k}.
\end{multline*}
The estimate \eqref{eq-kbd} is proved.

To show the equalities \eqref{eq-klim} let us study first the limits of $I_2$ and $I_3$.
By \eqref{eq-rbd} and due to the boundedness of $(a,b)$ one obtains by virtue of
the Lebesgue dominated convergence
\[
\lim_{\varepsilon\to 0+} I_2(\lambda,\varepsilon)=
\int_{a}^{b} \lim_{\varepsilon\to 0+} \dfrac{\varepsilon r(x,\lambda, \varepsilon)}{\lambda-m(x)+i\varepsilon m'(x)}\, dx,
\]
note that for $x$ satisfying $\lambda\ne m(x)$ (which can be violated for at most one point of
$[a,b]$) one has
\[
\lim_{\varepsilon\to 0+} \dfrac{\varepsilon r(x,\lambda, \varepsilon)}{\lambda-m(x)+i\varepsilon m'(x)}=0.
\]
Therefore, $\lim_{\varepsilon\to 0+} I_2(\lambda,\varepsilon)$. By the same arguments, 
$\lim_{\varepsilon\to 0+} I_3(\lambda,\varepsilon)$

To study the limit of $I_1$ we assume without loss of generality that $m'(x)>0$ on $I$ (otherwise one changes the signs of $T$, $m$ and $n$).
Introduce a new variable $y=m(x)$;
by the implicit function theorem one has $x=\varphi(y)$ and $\varphi'(y)=\big(m'(x)\big)^{-1}$.
This gives
\[
I_1(\lambda,\varepsilon)=
\dfrac{1}{\pi}\,\int_{m(a)}^{m(b)}
\dfrac{\varepsilon n\big(\varphi(y)\big)}{\big(\lambda-y\big)^2 + \dfrac{\varepsilon^2}{\varphi'(y)^2}}\,dy.
\]
Introducing another new variable $z=\dfrac{y-\lambda}{\varepsilon}$ one arrives at
\begin{equation}
           \label{eq-I1}
I_1(\lambda,\varepsilon)=
\dfrac{1}{\pi}\,\int_{\frac{m(a)-\lambda}{\varepsilon}}^{\frac{m(b)-\lambda}{\varepsilon}}
\dfrac{n\big(\varphi(\varepsilon z+\lambda)\big)}{z^2 + \dfrac{1\mathstrut}{\varphi'(\varepsilon z+\lambda)^2}}\,dy.
\end{equation}
One has
\[
\sup_{\frac{m(a)-\lambda}{\varepsilon}\le z\le \frac{m(b)-\lambda}{\varepsilon}}
\big|n\big(\varphi(\varepsilon z+\lambda)\big|=\sup_{a\le x\le b}\big|n(x)\big|\le N
\]
and 
\[
\inf_{\frac{m(a)-\lambda}{\varepsilon}\le z\le \frac{m(b)-\lambda}{\varepsilon}} \dfrac{1}{\varphi'(\varepsilon z+\lambda)^2}=
\inf_{a\le x\le b} m'(x)^2=k^2>0,
\]
therefore,
\[
\Bigg| \dfrac{n\big(\varphi(\varepsilon z+\lambda)\big)}{z^2 + \dfrac{1\mathstrut}{\varphi'(\varepsilon z+\lambda)^2}}\Bigg|
\le \dfrac{N}{z^2+\mu^2} \in L^1(\RR).
\]
Hence one has due to the Lebesgue dominated convergence
\[
\lim_{\varepsilon\to 0+} I_1(\lambda,\varepsilon)
=\dfrac{1}{\pi}\,\int_{\lim_{\varepsilon\to 0+}\frac{m(a)-\lambda}{\varepsilon}}^{\lim_{\varepsilon\to 0+}\frac{m(b)-\lambda}{\varepsilon}}
\lim_{\varepsilon\to 0+}\dfrac{n\big(\varphi(\varepsilon z+\lambda)\big)}{z^2 + \dfrac{1\mathstrut}{\varphi'(\varepsilon z+\lambda)^2}}\,dy.
\]
Recall that (for $a\ne 0$)
\[
\int_{-\infty}^{0}\dfrac{dt}{a^2+t^2}=\int_0^{+\infty}\dfrac{dt}{a^2+t^2}=\dfrac{1}{2}\int_{-\infty}^{+\infty}\dfrac{dt}{a^2+t^2}=\dfrac{\pi}{2|a|}.
\]
Clearly, for any $c\in\ J$
\[
\lim_{\varepsilon\to 0+}\frac{m(c)-\lambda}{\varepsilon}=\begin{cases}
+\infty, & \lambda<m(c)\\
0 & \lambda=m(c)\\
-\infty, & \lambda>m(c)
\end{cases}
\]
and that for $m(a)\le \lambda\le m(b)$ there holds
\[
\lim_{\varepsilon\to 0+}\dfrac{n\big(\varphi(\varepsilon z+\lambda)\big)}{z^2 + \dfrac{1\mathstrut}{\varphi'(\varepsilon z+\lambda)^2}}=
\dfrac{n\big(\varphi(\lambda)\big)}{z^2+\dfrac{1\mathstrut}{\varphi'(\lambda)^2}}.
\]
It remains to note that $\mu(x)=\varphi(x)$ for $x\in m(I\cap K)$.
The equalities \eqref{eq-klim} are hence obtained.
\end{proof}

\begin{lemma} \label{lem3a}
Let $L$ be a connected subset of $K$
such that $m(L)\cap\spec T\ne\emptyset$,
then the functions $m'$ and $n$ are either both strictly positive
on both strictly negative in $L$. 
\end{lemma}

\begin{proof}
Take $\lambda\in\spec T$ such that $\lambda\in m(L)$. 
As $\Im N(x+i\varepsilon)>0$ for $\varepsilon>0$, one has
\[
\dfrac{1}{2i}\,\Big(
\dfrac{n(x+i\varepsilon)}{\lambda-m(x+i\varepsilon)} - \dfrac{n(x-i\varepsilon)}{\lambda-m(x-i\varepsilon)}
\Big)\ge 0
\] 
for all $x\in\RR$.
Integrating this inequality on any $[a,b]\subset L$ such that
$\lambda\in m\big([a,b]\big)$ and passing to the limit as $\varepsilon\to 0+$
we obtain, by Lemma \ref{lem0}, $n\big(\mu(\lambda)\big)\mu'(\lambda)\ge 0$.
Let $\lambda=m(y)$, $y\in L$, then $0\le n\Big(\mu\big(m(y)\big)\Big)\mu'\big(m(y)\big)=\dfrac{n(y)}{m'(y)}$.
On the other hand, $n(y)\ne 0$ by assumption and $m'(y)\ne 0$ by Lemma \ref{prop-pm},
hence the inequality is strict, hence $m'(y)$ and $n(y)$ are either both negative or both positive.
As the two functions $m'$ and $n$ are continuous and do not vanish in the connected set $L$,
they have the same sign in whole $L$.
\end{proof}

Now we are able to show that $K$ has a rather simple structure given in Lemma \ref{lem4}.
\begin{proof}[Proof of Lemma \ref{lem4}] 
If the set $K$ is not connected,
then there are two different values $x_1,x_2\in J$ with $m(x_1)=m(x_2)=\tau$
with $\tau\in\big\{\inf\spec T,\sup\spec T\big\}$ (automatically $\tau\in\spec T$).
Due to analyticity of $m$ and without loss of generality one can
assume that $\tau=\sup\spec T$, that $x_1<x_2$ and that $m(x)>\tau$ for $x_1<x<x_2$.
Then $m'(x_1)>0$ and $m'(x_2)<0$. By the Lemma \ref{lem3a}, one has $n(x_1)>0$ and $n(x_2)<0$,
therefore, $n$ has to vanish in at least one point of the interval $(x_1,x_2)\subset J$,
which is impossible.
\end{proof}

Now we can prove the complete version of proposition \ref{prop-conv}.
\begin{proof}[Proof of Proposition \ref{prop-conv}]
By Lemma \ref{lem4}, there exists a bounded open interval $\Omega$
containing $m^{-1}(S_T)\cap J$ such that $m'(x)\ne 0$ for $x\in \Omega$.
Denote $L:=I\cap\Bar\Omega$ and $P:=\overline{I\setminus L}$.
One has $k_I(\lambda,\varepsilon)=k_P(\lambda,\varepsilon)+k_L(\lambda,\varepsilon)$.

Consider  the term $k_P$. As $m(P)\cap S_T=\emptyset$ by construction,
the subintegral expression
in \eqref{eq-kI} does not show any singularity for small $\varepsilon$,
i.e., for any $\varepsilon_0>0$ there exists $C>0$ such that
\[
\Big|\dfrac{n(x+i\varepsilon)}{\lambda-m(x+i\varepsilon)}-\dfrac{n(x-i\varepsilon)}{\lambda-m(x-i\varepsilon)}\Big| \le C
\]
for all $x\in P$, $\lambda\in S_T$ and $0<\varepsilon<\varepsilon_0$, and
\[
|k_P(\lambda,\varepsilon)|\le C|P| \text{ for all $\lambda\in S_T$ and $0<\varepsilon<\varepsilon_0$.}
\]
Futhermore, the Lebesgue dominated convergence and the equality
\[
\lim_{\varepsilon\to 0+}\dfrac{n(x+i\varepsilon)}{\lambda-m(x+i\varepsilon)}
=\lim_{\varepsilon\to 0+}\dfrac{n(x-i\varepsilon)}{\lambda-m(x-i\varepsilon)}
=\dfrac{n(x)}{\lambda-m(x)}
\]
implies $\lim_{\varepsilon\to 0+} k_P(\lambda,\varepsilon)=0 \text{ for all $\lambda\in S_T$}$.

To analyze the second term $k_L$, we remark that, by construction, $L$ is a closed interval and
$m'(x)\ne 0$ for $x\in L$, hence Lemma \ref{lem0} is applicable.
\end{proof}

\subsection{Spectral measures and proof of Theorem \ref{th-main}}

{}From now on we introduce the operator
\[
\widetilde T:=T_{m(J)}
\]
and the orthogonal projector
\[
P:\cG\to\widetilde \cG:=\ran E_T\big(m(J)\big)
\]
Recall that we consider $\widetilde T$ as a self-adjoint operator in $\widetilde\cG$.

\begin{prop}\label{prop2} 
Let $\mu$ be the inverse function to $K\ni x\mapsto m(x)\in m(K)\equiv m(J)$, then
the operator $n\big(\mu(\widetilde T)\big) \mu'(\widetilde T)$ is bounded, and
for any bounded Borel set $B\subset J$ there holds
\begin{gather}
 \Sigma_N(B)= P^* n\big(\mu(\widetilde T)\big) \mu'(\widetilde T) E_{\widetilde T}\big(m(B)\big)P, \label{eq-s0exp} \\
 \Sigma^0_N(B)=P^*n\big(\mu(\widetilde T)\big) \mu'(\widetilde T)\big(1+\mu(\widetilde T)^2\big)^{-1}E_{\widetilde T}\big(m(B)\big)P.
 \label{eq-sexp} 
\end{gather} 
\end{prop}

\begin{proof}
By the $\sigma$-additivity it is sufficient to consider open intervals $B=(a,b)$.

(a) Assume first $\Bar B=[a,b]\subset J$. Applying \eqref{eq-kbd} and
the Fubini theorem to the expression \eqref{eq-sigma0} for $\Sigma_0$
one obtains
\[
\Sigma_N(B)=\slim_{\delta\to 0+}\slim_{\varepsilon\to 0+} \int_{S_T} k_{[a+\delta,b-\delta]}(\lambda,\varepsilon)\, E_T(d\lambda).
\] 
Take any $h\in \cH$. Using again \eqref{eq-kbd} and the Lebesgue dominated convergence one obtains,
by virtue of \eqref{eq-klim},
\begin{multline}
\slim_{\varepsilon\to 0+} \int_{S_T} k_{[a+\delta,b-\delta]}(\lambda,\varepsilon)\, dE_T(\lambda)h\\
= \int_{S_T} \slim_{\varepsilon\to 0+} k_{[a+\delta,b-\delta]}(\lambda,\varepsilon)\, dE_T(\lambda)h%\\
=\widetilde f (T) E_T\Big(m \big((a+\delta,b-\delta)\big)\Big)h\\
+\dfrac{1}{2}\,\Big[\widetilde f\big(m(a+\delta)\big)E_T\big(\{m (a+\delta)\}\big)
%\\
+\widetilde{f}\big(m(b-\delta)\big)E_T\big(\{m (b-\delta)\}\big)
\Big]\,h
       \label{eq-limp}
\end{multline}
where
\[
\widetilde f(x)=\begin{cases}
n\big(\mu(x)\big)\mu'(x), & \text{for } x\in S_T\cap m(J),\\
0, & \text{otherwise}.
\end{cases}
\]
Hence, noting that the function $\widetilde f$ is a priori bounded on $m(B)$ and
passing to the limit as $\delta\to 0+$ we obtain 
\begin{equation}
          \label{eq-snb}
\Sigma_N(B):=\widetilde f(T) E_T\big(m(B)\big).
\end{equation}
On the other hand, there holds
\[
E_T\big(m(B)\big)=P^* E_{\widetilde T}\big(m(B)\big) P, \quad
\widetilde f(T):= P^* n\big(\mu(\widetilde T)\big)\mu'(\widetilde T)P,
\quad PP^*=\Id_{\widetilde G},
\]
which transforms \eqref{eq-snb} into \eqref{eq-s0exp}.

(b) Let $B=(a,b)\subset J$ be an arbitrary open interval. In this case the boundedness
of $\widetilde f$ on $m(B)$ is a priori not guaranteed, hence one can have troubles when passing to the limit
in \eqref{eq-limp}.
To deal with this case consider the sequence $B_n=(a+1/n,b-1/n)$. One has obvously $\Bar B_n\subset J$,
hence for any $h\in \dom L$, $L=\widetilde f(T)$, we have
\[
\lim_{n\to +\infty}E_T\big(m(B_n)\big) Lh=
E_T\big(m(B)\big)Lh.
\]
On the other hand, by (a), one has
\[
\slim_{n\to +\infty} L E_T\big(m(B_n)\big)=\slim_{n\to+\infty} \Sigma_N(B_n)=\Sigma_N(B).
\]
Therefore, for all $h\in\dom L$ we have
$LE_T\big(m(B)\big)h=\Sigma_N(B)h$, which is extended by continuity to all $h\in \cH$
and shows the boundedness of $L$.

(c) We have
\begin{multline*}
\Sigma^0_N(B)=\int_B \dfrac{\Sigma_N(dt)}{1+t^2}=
P^*\int_B \dfrac{n\big(\mu(\widetilde T)\big) \mu'(\widetilde T)E_{\widetilde T}\big(m(dt)\big)}{1+t^2}\, P\\
=P^* n\big(\mu(\widetilde T)\big) \mu'(\widetilde T) \int_{m(B)} \dfrac{E_{\widetilde T}(dy)}{1+\mu(y)^2}\, P\\
=P^* n\big(\mu(\widetilde T)\big) \mu'(\widetilde T)\big(1+\mu(\widetilde T)^2\big)^{-1}E_{\widetilde T}\big(m(B)\big)P.
\qedhere
\end{multline*}

\end{proof}

Now we are in position to conclude the proof of the main result.

\begin{proof}[Proof of Theorem \ref{th-main}]
 Recall that we have $R=\mu(\widetilde T)$,
and, therefore, $\widetilde T=m(R)$. Note first that the assertion (a) holds with $K$ defined in \eqref{eq-stk};
it satisfies the requested conditions due to Lemmas \ref{lem4} and \ref{lem3a}.

To proceed with the assertion (b), let us prove first the equality
\begin{equation}
         \label{eq-SNB0}
\Sigma_N(B)=P^* n(R)\big(m'(R)\big)^{-1}E_R\big(B\big)\, P^*
\text{ for all Borel sets $B\subset J$.}
\end{equation}

By the $\sigma$-additivity and the regularity arguments used in the proof of Proposition \ref{prop2}
it is sufficient to study the case when $B$ is an open interval such that $\Bar B\subset J$.
We have $E_{\widetilde T}\big(m(B)\big)=E_{m(R)}\big(m(B)\big)=E_R(B)$.
Substituting this equality in \eqref{eq-s0exp} and using
the identity $\mu'(x)=\big[m'\big(\mu(x)\big)\big]^{-1}$, we obtain
the requested equality \eqref{eq-SNB0}.
Analogously, from \eqref{eq-sexp} we deduce for $B\in \cB(\RR$), $B\subset J$,
\begin{equation}
         \label{eq-SNB}
\Sigma^0_N(B)=P^*n(R)\big(m'(R)\big)^{-1}(1+R^2)^{-1}E_R(B)P.
\end{equation}

Now consider the operator-valued measure $B\mapsto \Sigma^0_{N,J}(B):=\Sigma^0_{N}(B\cap J)$ on $\cG$.
One can rewrite \eqref{eq-SNB} as
\[
\Sigma^0_{N,J}(B)=D^* E_R (B) D,
\]
where 
\[
D=\Big[n(R)m'(R)^{-1}(1+R^2)^{-1}\Big]^{1/2}P.
\]
Note that the operator $n(R)m'(R)^{-1}$ is positive due to Lemma \ref{lem3a}, hence
$\ker D^*=0$ and $\overline{\ran D}=\widetilde \cG$.
Therefore, $\Sigma^0_{N,J}$ is a minimal dilation of the orthogonal measure $E_{R,J}$,
and the operators $H_J$ and $R$ are unitarily equivalent by Proposition \ref{prop-main}.
Theorem \ref{th-main} is proved.
\end{proof}

\section{Graph-like structures}\label{sec-gr}

In this section we are going to discuss a class of examples in which Weyl functions
of the form \eqref{eq-special} appear. We are interested in the case
$n\ne\text{const}$; examples with $n=\text{const}$ can be found e.g. in \cite[Section 4]{ABMN05}
or \cite[Subsection 1.4.4]{BGP08}. We introduce first a rather general abstract
construction and then discuss its realizations by quantum graphs.

\subsection{Gluing along graphs}\label{ss31}

A part of the constructions of this subsection already appeared in \cite{P08,P09}.
Let $G$ be a graph as in the introduction. For $v\in \cV$ we denote $E_v^\iota:=\{e\in\cE:\,\iota e=v\}\subset \cE$
and $E_v^\tau:=\{e\in\cE:\,\tau e=v\}\subset \cE$ and denote
by $E_v$ the \emph{disjoint} union of these two sets,
$E_v:=E_v^\iota\sqcup E_v^\tau$.

Let now $\cK$ be a Hilbert space and $L$ be a closed densely defined
symmetric operator in $\cK$ with the deficiency indices $(2,2)$.
Consider a boundary triple $(\CC^2,\pi,\pi')$ for $L$,
\[
\pi f =\begin{pmatrix}
\pi_\iota f \\ \pi_\tau f
\end{pmatrix},
\quad
\pi' f =\begin{pmatrix}
\pi'_\iota f \\ \pi'_\tau f
\end{pmatrix},
\]
and let $L^0$ be the restriction of $L^*$ to $\ker \pi$. Denote by $\gamma(z)$ the associated $\gamma$-field
and by $m(z)$ the corresponding Weyl function, which is in this case just
a $2\times 2$ matrix function,  
\[
m(z)=\begin{pmatrix}
m_{\iota\iota}(z) & m_{\iota\tau}(z)\\
m_{\tau\iota}(z) & m_{\tau\tau}(z)
\end{pmatrix}.
\] 

We are going to interpret the operator $L$ and its boundary triple
as a description of an object having two ends, $\iota$ and $\tau$, e.g. $\Gamma_\iota f$ and $\Gamma'_\iota f$
are interpreted as the boundary values of $f$ at $\tau$. Our aim is to replace each edge
of $G$ by a copy of this object and glue these copies together by
by suitable boundary conditions at the vertices. To make this construction more evident
and to provide it with a geometric interpretation let us consider two examples.

\begin{example}\label{ex-main}
Our main example is a Sturm-Liouville operator, see \cite[Section 4]{KP06} for the details of the construction.
 Let $l>0$
and let $V\in L^2(0,l)$ be a real-valued potential.
Consider the operator
\[
L:=-\dfrac{d^2}{dx^2}+V
\]
with the domain $H^2_0(0,l)=\{f\in H^2(0,l):\, f(0)=f(l)=f'(0)=f'(l)=0\}$.
Its adjoint $L^*$ is given by the same differential expression on the domain $H^2(0,l)$,
and as a boundary triple one can take
\begin{equation}
       \label{eq-btsl}
\pi f =\begin{pmatrix}
f(0) \\ f (l)
\end{pmatrix},
\quad
\pi'(f):=\begin{pmatrix}
f(0) \\ -f'(l)
\end{pmatrix}.
\end{equation}
The associated $\gamma$-field is given by
\[
\gamma(z)\begin{pmatrix} \xi_\iota \\ \xi_\tau
\end{pmatrix}
=\frac{\xi_e-\xi_\tau c(1;z)}{s(1;z)}\,s(x;z)+\xi_e c(x;z)
\]
and the Weyl function is
\begin{equation}
          \label{eq-msl}
m(z)=
\frac{1}{s(l;z)}\,\begin{pmatrix}
-c(l;z) & 1\\
1 & -s'(l;z)
\end{pmatrix},
\end{equation}
where $s$ and $c$ are the solutions of the differential equation $-y''(t)+V(t) y(t) =zy(t)$
satisying the boundary conditions $s(0;z)=c'(0;z)=0$ and $s'(0;z)=c(0;z)=1$.
Note that the associated operator $L^0$ is just the above Sturm-Liouville operator
with the Dirichlet boundary conditions at $0$ and $l$. Its spectrum $\sigma_D$ consists
of simple eigenvalues $\nu_n$, $n\in\NN$, $\nu_{n+1}>\nu_n$,
which are the zeros of the function $\nu\mapsto s(l;\nu)$.
\qed
\end{example}

\begin{example}\label{ex-lapl}
Let $L^0$ be the Laplace-Beltrami operator on a closed manifold $M$, $2\le \dim M\le 3$.
Take two points $x_1,x_2\in M$ and denote by $L$ the restriction of $L^0$
to the functions $f\in\dom L^0$ with $f(x_1)=f(x_2)=0$. Then $L$ is a closed symmetric operator
with deficiency indices $(2,2)$, and one can construct an associated boundary triple
and the Weyl function as follows, see \cite[Section 1.4.3]{BGP08}.
Let
\[
F(x,y)=\begin{cases}
\dfrac{1}{2\pi}\log\dfrac{1}{d(x,y)}, & \dim M=2,\\
\dfrac{1}{4\pi d(x,y)}, & \dim M=3,
\end{cases}
\]
where $d(x,y)$ is the geodesic distance between $x,y\in M$.
Any function $f\in \dom L^*$ has the asymptotic behavior
\[
f(x)= a_j(f) F(x,x_j) +b_j(f) +o(1),\quad x\to x_j, \quad a_j(f),b_j(f)\in\CC,\quad j=1,2,
\]
hence as a boundary triple one can take $(\CC^2,\Gamma,\Gamma')$
with
\[
\Gamma f=\begin{pmatrix}
a_1(f)\\ a_2(f)
\end{pmatrix},
\quad
\Gamma' f=\begin{pmatrix}
b_1(f)\\ b_2(f)
\end{pmatrix},.
\]
Note that the original operator $L^0$ is just the restriction
of $L^*$ to $\ker\Gamma$, and its spectrum is discrete.
The Weyl function $m$ for the above boundary triple has the form
\[
m(z)=\begin{pmatrix}
G^r(x_1,x_1;z) & G(x_1,x_2;z)\\
G(x_2,x_1;z) & G^r(x_2,x_2;z)
\end{pmatrix},
\]
where $G$ is the Green function of $L^0$, i.e. the integral kernel of the resolvent
$(L^0-z)^{-1}$, and $G^r$ is the regularized Green function,
defined as the difference $G^r(x,y;z):=G(x,y;z)-F(x,y)$
and extended to the diagonal $x=y$ by continuity.\qed
\end{example}

To introduce rigorously the gluing of copies of $L$ along the edges of $G$,
let us consider the Hilbert space $\cH:=\bigoplus_{e\in\cE} \cH_e$, $\cH_e= \cK$, and the symmetric operator
$S=\oplus_{e\in\cE} L_e$, $L_e= L$. Clearly, $S$ is closed densely defined in $\cH$, has equal deficiency indices,
and $S^*=\bigoplus_{e\in\cE} L_e^*$.
As a boundary triple for $S$ one can take $(\widetilde \cG,\widetilde \Gamma, \widetilde \Gamma')$
with
\[
\widetilde\cG:=\bigoplus_{e\in\cE} \CC^2, \quad
\widetilde \Gamma (f_e) = (\pi f_e), \quad
\widetilde \Gamma' (f_e) = (\pi' f_e).
\]
This construction does not take into account
the combinatorial structure of the graph $G$, and we prefer to modify it
by regrouping all the components with respect to the vertices.
More precisely, for any $v\in\cV$ denote $\cG_v:=\CC^{\deg v}$ and set $\cG:=\bigoplus_{v\in\cV} \cG_v$.
For $\phi\in\cG$ we will write $\phi=(\phi_v)_{v\in\cV}$, $\phi_v=(\phi_{v,e})_{e\in E_v}\in\cG_v$,
or simply $\phi=(\phi_{v,e})$. The scalar product of $\phi,\psi\in\cG$ is hence defined as
\[
\langle\phi,\psi\rangle_\cG=\sum_{v\in\cV} \langle \phi_v,\psi_v\rangle_{\cG_v}=
\sum_{v\in\cV}\sum_{e\in E_v} \overline{\phi_{e,v}}\psi_{e,v}.
\] 
As a boundary triple for $S$ we take now $(\cG,\Gamma,\Gamma')$
with
\[
\Gamma f= (\Gamma_v f)_{v\in\cV},\, \Gamma_v f=(\Gamma_{v,e} f)_{e\in E_v},\quad
\Gamma_{v,e}=\begin{cases}
\pi_\iota f_e & \text{if } v=\iota e,\\
\pi_\tau  f_e & \text{if } v=\tau e,
\end{cases}
\]
and $\Gamma'$ is defined analogously. Let us calculate the Weyl function
for this boundary triple. Let $\xi=(\xi_{v,e})\in\cG$ and $z\notin\spec L^0$.
The function $f\in\ker(S^*-z)$ with $\Gamma f= \xi$ has the form
$f=(f_e)$, 
\[
f_e=\gamma(z) \begin{pmatrix}
\xi_{\iota e,e}\\ \xi_{\tau e,e}
\end{pmatrix}, 
\quad
\begin{pmatrix}
\Gamma'_{\iota e,e} f\\
\Gamma'_{\tau e,e} f
\end{pmatrix}
=\pi' \gamma(z)\begin{pmatrix}
\xi_{\iota e,e}\\ \xi_{\tau e,e}
\end{pmatrix}= m(z)\begin{pmatrix}
\xi_{\iota e,e}\\ \xi_{\tau e,e}
\end{pmatrix}.
\]
Therefore,
\begin{equation}
       \label{eq-mzz}
\big(M(z) \xi\big)_{v,e}=\Gamma'_{v,e}f=\begin{cases}
m_{\iota\iota}(z) \xi_{v,e} + m_{\iota\tau}(z) \xi_{v_e,e}, & \text{if } v=\iota e,\\
m_{\tau\tau}(z) \xi_{v,e} + m_{\tau\iota}(z) \xi_{v_e,e}, & \text{if } v=\tau e,
\end{cases} 
\end{equation}
where
\[
v_e=\begin{cases}
\tau e & \text{for } v=\iota e,\\
\iota e & \text{for } v=\tau e.
\end{cases}
\]
Note that if the symmetry conditions
\begin{equation}
   \label{eq-sym}
m_{\iota\iota}(z)=m_{\tau\tau}(z) \text{ and } m_{\iota\tau}(z)=m_{\tau\iota}(z)
\end{equation}
are satified, then the above expression for $M(z)$ can be simplified to
\begin{equation}
   \label{eq-mzsym}
   M(z) = m_{\iota\iota}(z)\Id+  m_{\iota\tau}(z)\,D,
\end{equation} 
where $D$ is the self-adjoint operator in $\cG$ acting as
\[
\big(D \xi\big)_{v,e}= \xi_{v_e,e}.
\]
The restriction $H^0$ of $S^*$ to $\ker \Gamma$ is just the direct sum
of the copies of $L^0$,
\[
H^0=\bigoplus_{e\in\cE} L^0,
\]
hence $\spec H^0=\spec L^0$ and any spectral gap of $L^0$
is also a spectral gap for $H^0$.

Now impose gluing boundary conditions at each vertex $v\in\cV$ by
\begin{equation}
       \label{eq-bcab}
A_v \Gamma_v f=B_v\Gamma'_v f
\end{equation}
where $A_v$, $B_v$ are ${\deg v}\times{\deg v}$ matrices such that $A_v B_v^*=B_v A_v^*$ and $\det(A_vA_v^*+B_v B_v^*)>0$.
One can rewrite these conditions in the equivalent normalized form
\begin{equation}
      \label{eq-bcu}
(1-U_v)\Gamma_v =i(1+U_v)\Gamma'_v f, \quad U_v\in \cU(\deg v)
\end{equation}
or
\begin{equation}
      \label{eq-bcpc}
P_v \Gamma'_v f = C_v P \Gamma_v f, \quad (1-P_v)\Gamma_v f=0,
\end{equation}
where $P_v$ is the orthogonal projector from $\CC^{\deg v}$ to 
\[
\cL_v:=\ker(1+U_v)^\perp
\]
and $C_v$ is a self-adjoint operator in $\cL_v$ defined as
\[
C_v=-i (1-P_v U_v P_v^*)(1+P_vU_vP_v^*)^{-1}.
\]
The equivalent boundary conditions \eqref{eq-bcab}, \eqref{eq-bcu}, \eqref{eq-bcpc} define
a self-adjoint operator, see e.g. \cite[Section 1]{BGP08}, and we denote this operator by $H$.
Note that in general $H$ is not disjoint with $H^0$ as one has $\dom H\cap \dom H^0=\ker P\Gamma'\cap\ker\Gamma\ne \dom S$,
$P:=\bigoplus_{v\in\cV} P_v$, so let us proceed as in \cite[Theorem 1.32]{BGP07}.

Denote by $\widetilde S$ the restriction of $S^*$ to
$\ker P\Gamma'\cap\ker\Gamma$, then $\widetilde S^*$ is the restriction of $S^*$ to $\ker(1-P)\Gamma$,
and as a boundary triple for $\widetilde S$
one can take $(\cG_P,\Gamma_P,\Gamma'_P)$ defined
by 
\[
\cG_P=\ran P=\bigoplus_{v\in\cV} \cL_v,
\quad
\Gamma_P=P\Gamma P^*,
\quad
\Gamma'_P:=P\Gamma'P^*
\]
($\cG_P$ is considered with the scalar product induced by the inclusion $\cG_P\subset\cG$),
and the associated Weyl function $M_P$ takes the form
\[
M_P(z):=P M(z)P^*.
\]
Now $H$ becomes the restriction
of $\widetilde S^*$ to the vectors $f$ satisfying
\[
\Gamma'_P f:= C \Gamma_P f, \quad C:= \bigoplus_{v\in\cV}C_v,
\]
and the operator $H^0$ is still the restriction of $\widetilde S^*$ to $\ker\Gamma_P$.
The following theorem shows that the spectral analysis of $H$ can be reduced in certain cases
to the spectral analysis of the discrete operator $D_P$ on $\cG_P$,
\[
D_P:=P D P^*.
\]

\begin{theorem}\label{th-red}
Assume that the symmetry conditions \eqref{eq-sym} hold and that there is $\theta \in\CC$,
such that $|\theta|=1$, $\theta\ne -1$, and
\begin{equation}
   \label{eq-scal}
\bigcup_{v\in\cV} \spec U_v \setminus \{-1\}=\{\theta\},
\end{equation}
Set
\[
\alpha:=-\frac{i\,(1-\theta)}{1+\theta}, \quad
\eta_\alpha(z):= \frac{\alpha -m_{\iota\iota}(z)}{m_{\iota\tau}(z)}.
\]
Assume now that there exists an interval $J\subset\RR\setminus\spec L^0$ such that
$m_{\iota\tau}(z)\ne 0$ for $z\in J$.
Then the operators $H_J$ and $\eta_\alpha^{-1}\big((D_P)_{\eta_\alpha(J)}\big)$
are unitarily equivalent.
\end{theorem}

\begin{proof}
Let us show that the assumptions of Theorem \ref{th-main} are satisfied.
First of all, as mentioned above, due to \eqref{eq-sym} and  \eqref{eq-mzsym}
one has $M_P(z):=m_{\iota\iota}(z)\Id_P +m_{\iota\tau}(z) D_P$.
On the other hand, under the assumption \eqref{eq-scal} all the operators $C_v$ are just
the multiplications by $\alpha$,
hence $H$ is the restriction of $\widetilde S^*$ to $\ker(\Gamma'_P-\alpha\Gamma_P)$.
Now introduce another boundary triple $(\cG_P,\Gamma_{P,\alpha},\Gamma'_{P,\alpha})$ for $\widetilde S$
by $\Gamma_{P,\alpha}=\Gamma_P$ and $\Gamma'_{P,\alpha}=\Gamma'_P-\alpha \Gamma_P$.
The associated Weyl function is
\[
M_{P,\alpha}(z)=M_P(z)-\alpha \Id= \big(m_{\iota\iota}(z) -\alpha\big)\Id + m_{\iota\tau} D_P=
\dfrac{\eta_\alpha(z)\,\Id-D_P}{-m_{\iota\tau}(z)^{-1}}.
\]
As $H=\widetilde S^*_{\ker\Gamma'_{P,\alpha}}$, the result
follows from Theorem \ref{th-main}.
\end{proof}

In the example \ref{ex-main}, the symmetry conditions \eqref{eq-sym} are satisfied if the potential $V$
is symmetric, i.e. if $V(x)\equiv V(l-x)$, cf. \cite[Section 4]{KP06}. In the example \ref{ex-lapl}
these conditions hold, e.g. if there exists an isometry $g$ of $M$ such that $g(x_1)=x_2$. If $M$ is a two-dimensional
sphere, then the condition \eqref{eq-sym} holds for arbitrary $x_1$ and $x_2$;
we refer to the paper~\cite{BEG} studying various systems of coupled spheres.
Note also that the operator $D_P$ can be viewed as a generalized laplacian
on the graph $G$, see \cite{P08,P09}. We will also see below that the adjacency
operator \eqref{eq-delta} is a particular case of $D_P$ for a suitable projector $P$.

\subsection{Quantum graph case}

Consider now in greater detail the constructions of subsection \ref{ss31}
for the Sturm-Liouville operator $L$ from Example \ref{ex-main}.

Let, as previously, $l>0$, $V\in L^2(0,l)$ be a real-valued potential and fix $\alpha:\cV\to \RR$. 
Denote by $H$ the self-adjoint operator acting in $\cH:=\bigoplus_{e\in\cE} L^2(0,l)$ as
$(f_e)\mapsto (-f_e''+Vf_e)$ on the functions $f=(f_e)\in \bigoplus_{e\in\cE}H^2(0,l)$
satisfying the boundary conditions
\begin{equation}
  \label{eq-fv}
\begin{gathered}
\text{the value } f_e(v)=: f(v) \text{ is the same for all } e\in E_v,\\
\sum_{e:\iota{e}=v} f'_e(v)=\alpha(v) f(v), \quad v\in\cV,
\end{gathered}
\end{equation}
where we denote
\[
f_e(v)=\begin{cases}
f_e(0) & \text{if } \iota e=v,\\
f_e(l) & \text{if } \tau e=v,
\end{cases}
\quad
f'_e(v)=\begin{cases}
f'_e(0) & \text{if } \iota e=v,\\
-f'_e(l) & \text{if } \tau e=v.
\end{cases}
\]
Recall that by $\sigma_D$ we denote the spectrum of the operator $f\mapsto -f''+Vf$
on $[0,l]$ with the Dirichlet boundary conditions.

The operator $H$ has the structure requested in subsection \ref{ss31}: it represents copies
of the same operator $L$ from Example \ref{ex-main} coupled through boundary conditions
at each vertex of the graph. One can rewrite the boundary conditions \eqref{eq-fv}
in the normalized form \eqref{eq-bcu} with 
\[
U_v=\dfrac{2}{\deg v + i \alpha(v)}\, J_{\deg v} -I_{\deg v},
\]
here $I_n$ and $J_n$ are respectively the $n\times n$ identity matrix and the $n\times n$ matrix whose all entries are $1$
\cite{CE}.
The value $-1$ is an eigenvalue of $U_v$ of multiplicity $\deg v -1$, and the orthogonal
projector $P_v$ onto $\ker(U_v+1)^\perp$ is just the orthogonal projector onto the one-dimensional space
spanned by the vector $p_v$, where $p_v$ is the vector of length $\deg v$ whose all entries are $1$, i.e.,
in the matrix form,
\[
P_v=\dfrac{1}{\deg v} J_{\deg v}Á
\]
Finally let us note that the condition \eqref{eq-scal} is satisfied if one has
\begin{equation}
        \label{eq-av}
\alpha(v)=\alpha \deg v
\end{equation}
for some $\alpha\in\RR$. Theorem \ref{th-red} applied to the case under consideration gives
\begin{theorem}\label{th2}
Assume that the potential $V$ is symmetric, $V(x)\equiv V(l-x)$, and that the condition \eqref{eq-av} holds.
Then, for any interval $J\subset\RR\setminus \sigma_D$
the operator $H_J$ is unitarily equivalent to $\eta_{\alpha}^{-1}\big(\Delta_{\eta_\alpha(J)}\big)$,
where $\Delta$ is the operator in $l^2(G)$ given by \eqref{eq-delta} and
\begin{equation}
       \label{eq-etaa}
\eta_\alpha (z) = c(l;z)+\alpha s(l;z).
\end{equation}
\end{theorem}

\begin{proof}
As noted above, the symmetry of the potential $V$ guarantees that the conditions \eqref{eq-sym} hold.
Theorem \ref{th-red} and the formulas \eqref{eq-msl} 
show that $H_J$ is unitarily equivalent
to $\eta_{\alpha}^{-1}\big((D_P)_{\eta_\alpha(J)}\big)$ 
On the other hand, consider the unitary transformation
\begin{equation}
     \label{eq-theta}
\Theta:l^2(G)\to \cG_P, \quad (\Theta \xi)_v= \xi(v) p_v. 
\end{equation}
Applying $D_P$ to $\Theta\xi$ we obtain
\begin{multline*}
( D_P \Theta\xi)_{v,e}=
(P D P^* \Theta \xi)_{v,e}=\dfrac{1}{\deg v} \sum_{e\in E_v} \big(D P^* \Theta\xi\big)_{v,e}\\
=\dfrac{1}{\deg v} \sum_{e\in E_v} (\Theta \xi)_{v_e,e}=\dfrac{1}{\deg v}\sum_{e\in E_v} \xi(v_e), 
\end{multline*}
i.e. $D_P \Theta=\Theta\Delta$, hence $D_P$ and $\Delta$ are unitarily equivalent.
\end{proof}

Taking in this theorem $l=1$, $V=0$ and $\alpha=0$
we obtain $\eta_0(z)=\cos\sqrt z$, which gives proposition \ref{th0}.

Let us mention several other cases where the unitary dimension reduction is possible.

\begin{theorem}\label{th3}
Let $V\in L^2(0,l)$ be arbitrary and the condition \eqref{eq-av} hold.
Assume that the ratio $\kappa:=\dfrac{\outdeg v}{\deg v}$
is the same for all $v\in\cV$. Then $H_J$ is unitarily equivalent to
$\eta_{\alpha}^{-1}\big(\Delta_{\eta_\alpha(J)}\big)$ with
$\eta_\alpha(z)= \kappa c(l;z)+(1-\kappa)s'(l;z)+\alpha s(l;z)$. 
\end{theorem}

\begin{proof}
Note that we still have $m_{\iota\tau}=m_{\tau\iota}$.
Take the same unitary transformation \eqref{eq-theta} and calculate $M_P \Theta$:
\begin{multline*}
(PM(z)P^* \Theta)\xi_{v,e}=\dfrac{1}{\deg v}\Big\{
\sum_{e:\iota e=v} \big[m_{\iota\iota}(z) (\Theta \xi)_{v,e} - m_{\iota\tau}(z) (\Theta \xi)_{v_e,e}\big]\\
+\sum_{e:\tau e=v} \big[m_{\tau\tau}(z) (\Theta \xi)_{v,e} - m_{\tau\iota}(z) (\Theta \xi)_{v_e,e}\big]
\Big\}\\
= \dfrac{1}{\deg v}\,\Big[ \big(\outdeg v  \cdot  m_{\iota\iota}(z) +\indeg v  \cdot  m_{\iota\iota}(z)\big)
\xi(v)+ m_{\iota\tau}(z)\sum_{e\in E_v} \xi(v_e)
\Big],
\end{multline*}
hence
\[
M_P(z)\Theta=\dfrac{\Theta\Delta -\big(\kappa c(l;z)+ (1-\kappa) s'(l;z)\big)\Theta}{s(l;z)},
\]
and the rest of the proof is similar to that of Theorem \ref{th-red}.
\end{proof}

One can extend the above results to the case with magnetic fields following the constructions
of \cite{KP06,P09}. Namely, let $(a_e)_{e\in\cE}$
be a family of magnetic potentials, $a_e\in C^1 \big([0,l\big)]$.
Denote by $\widetilde H$ the self-adjoint operator in $\cH:=\bigoplus_{e\in\cE} L^2(0,l)$ as
\[
(g_e)\mapsto \Big((i\partial +a_e)^2g_e''+Vg_e\Big), \quad \partial g_e:=g'_e, 
\]
on the functions $g=(g_e)\in \bigoplus_{e\in\cE}H^2(0,l)$
satisfying the magnetic analogue of the boundary conditions \eqref{eq-fv},
\begin{gather*}
\text{the value } g_e(v)=: g(v) \text{ is the same for all } e\in E_v,\\
\sum_{e:\iota{e}=v} \big[g'_e(v)-ia_e(v) g_e(v)\big]=\alpha(v) g(v), \quad v\in\cV.
\end{gather*}
Applying the unitary transformation
\[
g_e(t)=\exp\Big(\int_0^t a_e(s)ds\Big) f_e(t)
\]
and introducing the parameters
\[
\beta_e=\int_0^l a_e(s)ds
\]
one sees that $\widetilde H$ is unitarily equivalent to the operator $H$ acting as
$(f_e)\mapsto(-f_e''+V f_e)$ with the boundary conditions
\begin{gather*}
\text{the value } e^{i\beta_{v,e}}f_e(v)=: f(v) \text{ is the same for all } e\in E_v,\\
\sum_{e:\iota{e}=v} e^{i\beta_{v,e}}f'_e(v)=\alpha(v) g(v), \quad v\in\cV, \quad
\text{with } \beta_{v,e}=\begin{cases}
0 & \text{ if } v=\iota e,\\
\beta_e & \text{ if } v=\tau e.
\end{cases}
\end{gather*}
By a minor modification of the preceding constructions one can show that Theorems \ref{th2} and \ref{th3}
hold in the same form if one replaces the operator $\Delta$ by its magnetic version $\Delta_\beta$,
\[
\Delta_\beta f(v)=\frac{1}{\deg v}\,\Big(\sum_{e:\iota e=v} e^{-i\beta_e}f(\tau e)
+\sum_{e:\tau e=v} e^{i\beta_e}f(\iota e)\,\Big).
\]

Let us now comment on the dimension reduction for boundary conditions
different from \eqref{eq-fv}.

\begin{example}[$\delta'$-coupling]\label{ex-dprim}
Another popular class of boundary conditions is the so-called
$\delta'$ coupling \cite{CE},
\[
\sum_{e\in E_v} f'_e(v)=0,\quad
f_e(v)-f_b(v)=\dfrac{\beta(v)}{\deg v}\big(
f'_e(v)-f'_b(v)\big), \quad e,b\in E_v, \quad v\in\cV,
\]
where $\beta(v)$ are non-zero real constants.
These boundary conditions can be rewritten in the normalized form \eqref{eq-bcu}
with
\[
U(v)=-\dfrac{\deg v+i\beta(v)}{\deg v-i\beta(v)}\, I_{\deg v} + \dfrac{2}{\deg v-i\beta(v)}\,J_{\deg v},
\]
and the condition \eqref{eq-scal} is fulfilled if $\beta(v)=\beta \deg v$ for some $\beta\in\RR\setminus\{0\}$.
Hence for an even potential $V$ Theorem \ref{th-red} applies, and for any
interval $J\subset\RR\setminus\sigma_D$ the operator $H_J$ is unitarily equivalent to $\eta_{1/\beta}^{-1}\big((D_P)_{\eta_{1/\beta}(J)}\big)$
with $\eta_{1/\beta}$ defined by \eqref{eq-etaa} and $P=\bigoplus P_v$,
where $P_v$ is the orthogonal projector in $\CC^ {\deg v}$ onto the subspace $p_v^\perp$.
Such operator $D_P$ appeared already in \cite{MH} in a slightly different problem. \qed
\end{example}

\begin{example}[$\delta'_s$ coupling] One can also consider the so-called $\delta'_s$ coupling
given by the following boundary conditions \cite{CE}:
\begin{equation}
        \label{eq-dps}
f'_e(v)=f'_b(e)=:f'(v), \quad e,b\in E_v, \quad
\sum_{e\in E_v} f_e(v)=\alpha(v) f'(v), \quad v\in\cV.
\end{equation}
To treat this case it is better to modify the boundary triple for the initial operator
$L$: instead of \eqref{eq-btsl} one can define
\[
\pi f=\begin{pmatrix}
-f'(0)\\ f'(l)
\end{pmatrix},
\quad
\pi' f=\begin{pmatrix}
f(0)\\ f(l)
\end{pmatrix},
\]
then the associated Weyl function is
\[
m(z)=\dfrac{1}{c'(l;z)}\,\begin{pmatrix}
s'(l;z) & 1 \\ 1 & c(l;z)
\end{pmatrix}.
\]
Note that the reference operator $L^0$ is now the Neumann operator on $[0,l]$.
Denote by $\sigma_N$ its spectrum. With this new boundary triple
the boundary conditions \eqref{eq-dps} become similar to the Kirchoff boundary conditions \eqref{eq-fv};
they can rewritten in the normalized form \eqref{eq-bcu} with
\[
U_v=\dfrac{1}{\deg v - i \alpha(v)}\,J_{\deg v}- I_{\deg v}.
\]
Assuming now that $V$ is symmetric and that \eqref{eq-scal} holds and proceeding as in Theorem \ref{th2}
one can show that for any interval $J\subset\RR\setminus\sigma_N$ the operator $H_J$
is unitarily equivalent to $\eta_\alpha^{-1}\big((-\Delta)_{\eta_\alpha(J)}\big)$
with $\eta_\alpha(z)=c(l;z)+\alpha c'(l;z)$. \qed
\end{example}

\end{document}